\RequirePackage{fix-cm}
\documentclass[smallextended]{svjour3}       % onecolumn (second format)
\smartqed  % flush right qed marks, e.g. at end of proof

\usepackage{amssymb,amsmath,amsfonts,mathrsfs}
\usepackage[utf8]{inputenc}
\usepackage[english]{babel}
\usepackage{graphicx}
\usepackage{color}
\usepackage{doi}

\DeclareMathOperator{\length}{length}
\DeclareMathOperator{\rank}{rank}
\DeclareMathOperator{\cone}{cone}
\DeclareMathOperator{\conv}{conv\!.\!hull}

\DeclareMathOperator{\inth}{\Lambda}

\DeclareMathOperator{\vertex}{vert}
\DeclareMathOperator{\lcm}{lcm}

\DeclareMathOperator{\poly}{poly}

\DeclareMathOperator{\relint}{rel\!.\!int}
\DeclareMathOperator{\inter}{int}

\DeclareMathOperator{\diag}{diag}

\DeclareMathOperator{\tcone}{tcone}
\DeclareMathOperator{\paral}{\Pi}
\DeclareMathOperator{\fcone}{fcone}
\DeclareMathOperator{\lmod}{\text{\it\quad modulo polyhedra with lines}}

\DeclareMathOperator{\ldcmod}{\text{\it\quad modulo lower-dimensional rational cones}}
\DeclareMathOperator{\toddp}{td}

\DeclareMathOperator{\BUnit}{\mathbf 1}
\DeclareMathOperator{\BZero}{\mathbf 0}
\DeclareMathOperator{\BX}{\mathbf x}
\DeclareMathOperator{\BY}{\mathbf y}

\DeclareMathOperator{\TC}{\mathcal{T}}
\DeclareMathOperator{\ZZ}{\mathbb{Z}}
\DeclareMathOperator{\QQ}{\mathbb{Q}}
\DeclareMathOperator{\RR}{\mathbb{R}}
\DeclareMathOperator{\CC}{\mathbb{C}}
\DeclareMathOperator{\EC}{\mathcal{E}}
\DeclareMathOperator{\PC}{\mathcal{P}}
\DeclareMathOperator{\DC}{\mathcal{D}}

\DeclareMathOperator{\BC}{\mathcal{B}}
\DeclareMathOperator{\NN}{\mathbb{N}}
\DeclareMathOperator{\NC}{\mathcal{N}}
\DeclareMathOperator{\IP}{ILP}
\DeclareMathOperator{\LP}{LP}

\DeclareMathOperator{\FC}{\mathcal{F}}

\DeclareMathOperator{\SNF}{\text{\rm SNF}}

\newcommand*{\intint}[2][1]{#1\!:\!#2}

\begin{document}

\title{On lattice point counting in $\Delta$-modular polyhedra%\thanks{The article was prepared within the framework of the Basic Research Program at the National Research University Higher School of Economics (HSE).}
}
\author{D.~V.~Gribanov, N.~Yu.~Zolotykh}
\institute{D.~V.~Gribanov \at National Research University Higher School of Economics, 25/12 Bolshaja Pecherskaja Ulitsa, Nizhny Novgorod, 603155, Russian Federation\\
\email{dimitry.gribanov@gmail.com}
\and  N.~Yu.~Zolotykh \at Mathematics of Future Technologies Center, Lobachevsky State University of Nizhni Novgorod, 23 Gagarin ave, Nizhni Novgorod, 603950, Russian Federation\\
\email{nikolai.zolotykh@itmm.unn.ru}
%\and  P.~M.~Pardalos \at University of Florida, 401 Weil Hall,
%P.O. Box 116595, Gainesville, FL 326116595, USA\\
%\email{pardalos@ufl.edu}
}
\maketitle              % typeset the header of the contribution
\begin{abstract}
Let a polyhedron $P$ be defined by one of the following ways: 
\begin{enumerate}
\item[(i)] $P = \{x \in \RR^n \colon A x \leq b\}$, where $A \in \ZZ^{(n+k) \times n}$, $b \in \ZZ^{(n+k)}$ and $\rank A = n$,
\item[(ii)] $P = \{x \in \RR_+^n \colon A x = b\}$, where $A \in \ZZ^{k \times n}$, $b \in \ZZ^{k}$ and $\rank A = k$,
\end{enumerate} 
and let all rank order minors of $A$ be bounded by $\Delta$ in absolute values. We show that the short rational generating function for the power series 
$$
\sum\limits_{m \in P \cap \ZZ^n} \BX^m
$$
can be computed with the arithmetical complexity 
$
O\left(T_{\SNF}(d) \cdot d^{k} \cdot d^{\log_2 \Delta}\right),
$
where $k$ and $\Delta$ are fixed, $d = \dim P$, and $T_{\SNF}(m)$ is the complexity of computing the Smith Normal Form for $m \times m$ integer matrices.
In particular, $d = n$, for the case (i), and $d = n-k$, for the case (ii).

The simplest examples of polyhedra that meet the conditions (i) or (ii) are the \emph{simplices}, the \emph{subset sum} polytope and the \emph{knapsack} or \emph{multidimensional knapsack}  polytopes. Previously, the existence of a polynomial time algorithm in varying dimension for the considered class of problems was unknown already for simplicies ($k = 1$).

We apply these results to parametric polytopes and show that the step polynomial representation of the function $c_P(\BY) = |P_{\BY} \cap \ZZ^n|$, where $P_{\BY}$ is a parametric polytope, whose structure is close to the cases (i) or (ii), can be computed in polynomial time even if the dimension of $P_{\BY}$ is not fixed. As another consequence, we show that the coefficients $e_i(P,m)$ of the Ehrhart quasi-polynomial 
$$
\left|  mP \cap \ZZ^n\right| = \sum\limits_{j = 0}^n e_j(P,m)m^j
$$ 
can be computed with a polynomial-time algorithm, for fixed $k$ and $\Delta$. 

%Using the lemma about minors of orthogonal matricies we translate these results to polyhedrons defined by systems in \emph{ the standard form} $P = \{x \in \RR^n_+ \colon A x = b\}$, where $A \in \ZZ^{k \times n}$ and $b \in \ZZ^k$. Resulting algorithms have the arithmetic complexity 
%$$
%O\left(d^{2k + O(1)} + T_{\SNF}(d) \cdot d^{k} \cdot d^{\log_2 \Delta}\right),
%$$ where $d = \dim P = n - k$. 

% Additionally, we do some notes on counting all solutions of the unbounded knapsack, subset sun and multidimensional unbounded knapsack polytopes in the case, when weights are bounded by some constant.

\end{abstract}

\section{Introduction}

Let a polyhedron $P$ be defined with one of the following ways: 
\begin{enumerate}
\item[(i)] $P = \{x \in \RR^n \colon A x \leq b\}$, where $A \in \ZZ^{(n+k) \times n}$, $b \in \ZZ^{(n+k)}$ and $\rank A = n$,
\item[(ii)] $P = \{x \in \RR_+^n \colon A x = b\}$, where $A \in \ZZ^{k \times n}$, $b \in \ZZ^{k}$ and $\rank A = k$.
\end{enumerate} 
The simplest examples of polytopes that can be represented this way are the \emph{simplices}, the \emph{subset sum} polytope, and the \emph{knapsack} and \emph{multidimensional knapsack}  polytopes.

Let all rank order minors of $A$ be bounded by $\Delta$ in absolute values. We show that \emph{the short rational generating function} for the power series 
$$
\sum\limits_{m \in P \cap \ZZ^n} \BX^m
$$
can be computed with the arithmetical complexity 
$
O\left(T_{\SNF}(d) \cdot d^{k} \cdot d^{\log_2 \Delta}\right),
$
where $T_{\SNF}(m)$ is the complexity of computing \emph{the Smith Normal Form} for $m \times m$ integer matrices and 
$d = \dim P$; in particular, $d = n$, for the case (i), and $d = n-k$, for the case (ii). The complexity bound is polynomial for fixed $k$ and $\Delta$. Previously, the existence of a polynomial time algorithm in varying dimension for the considered class of problems was unknown already for simplicies ($k = 1$).

We apply these results to parametric polytopes, and show that the step polynomial representation of the function $c_P(\BY) = |P_{\BY} \cap \ZZ^n|$, where $P_{\BY}$ is a parametric polytope, whose structure is close to the cases (i) or (ii), can be computed in polynomial time even if the dimension of $P_{\BY}$ is not fixed. As another simple consequence, we show that the coefficients $e_i(P,m)$ of the \emph{Ehrhart quasi-polynomial} 
$$
\left|  mP \cap \ZZ^n\right| = \sum\limits_{j = 0}^n e_j(P,m)m^j
$$ 
can be computed with a polynomial-time algorithm for fixed $k$ and $\Delta$.

%Using the lemma about minors of orthogonal matricies we translate these results to polyhedrons defined by systems in \emph{ the standard form} $P = \{x \in \RR^n_+ \colon A x = b\}$, where $A \in \ZZ^{k \times n}$ and $b \in \ZZ^k$. Resulting algorithms have the arithmetic complexity 
%$$
%O\left(d^{2k + O(1)} + T_{\SNF}(d) \cdot d^{k} \cdot d^{\log_2 \Delta}\right),
%$$ where $d = \dim P = n - k$. 

% Additionally, we do some notes on counting all solutions of the \emph{unbounded knapsack}, \emph{subset sun} and \emph{multidimensional unbounded knapsack} polytopes in the case, when weights are bounded by some constant.

Our method is based on the approach developed by A.~Barvinok \cite{BARV93,BARVBOOK,BARVPOM,BARVWOODS},  but we use another variant of the sign decomposition technique.

\section{Basic definitions and notations}

Let $A \in \mathbb{Z}^{m \times n}$ be an integer matrix. We denote by $A_{ij}$ the $ij$-th element of the matrix, by $A_{i*}$ its $i$-th row, and by $A_{*j}$ its $j$-th column. The set of integer values from $i$ to $j$, is denoted by $\intint[i]j=\left\{i, i+1, \ldots, j\right\}$. Additionally, for subsets $I \subseteq \{1,\dots,m\}$ and $J \subseteq \{1,\dots,n\}$, the symbols $A_{IJ}$ and $A[I,J]$ denote the sub-matrix of $A$, which is generated by all the rows with indices in $I$ and all the columns with indices in $J$. If $I$ or $J$ are replaced by $*$, then all the rows or columns are selected, respectively. Sometimes, we simply write $A_{I}$ instead of $A_{I*}$ and $A_{J}$ instead of $A_{*J}$, if this does not lead to confusion.

The maximum absolute value of entries in a matrix $A$ is denoted by $\|A\|_{\max} = \max_{i,j} |A_{i\,j}|$. The $l_p$-norm of a vector $x$ is denoted by $\|x\|_p$. The number of non-zero items in a vector $x$ is denoted by $\|x\|_0 = |\{i\colon x_i \not= 0\}|$. The column compounded by diagonal elements of a $n \times n$ matrix $A$ is denoted by $\diag(A) = (A_{11},\dots,A_{nn})^\top$. Its adjugate matrix is denoted by $A^* = \det(A) A^{-1}$.

\begin{definition}
For a matrix $A \in \ZZ^{m \times n}$, by $$
\Delta_k(A) = \max\{|\det A_{IJ}| \colon I \subseteq \intint m,\, J \subseteq \intint n,\, |I| = |J| = k\},
$$ we denote the maximum absolute value of determinants of all the $k \times k$ sub-matrices of $A$.  By $\Delta_{\gcd}(A,k)$ and $\Delta_{\lcm}(A,k)$, we denote the greatest common divisor and the least common multiplier of nonzero determinants of all the $k \times k$ sub-matrices of $A$, respectively. Additionally, let $\Delta(A) = \Delta_{\rank A}(A)$, $\Delta_{\gcd}(A) = \Delta_{\gcd}(A,\rank A)$, and $\Delta_{\lcm}(A) = \Delta_{\lcm}(A, \rank A)$.
\end{definition}

\begin{definition}
For a matrix $B \in \RR^{m \times n}$,
$\cone(B) = \{B t\colon t \in \RR_+^{n} \}$ is the \emph{cone, spanned by columns of} $B$,
$\conv(B) = \{B t\colon t \in \RR_+^{n},\, \sum_{i=1}^{n} t_i = 1  \}$ is the \emph{convex hull, spanned by columns of} $B$,
%$\affh(B) = \{B t\colon t \in \RR^n,\, \sum_{i=1}^{n} t_i = 1\}$ is the \emph{affine hull spanned by columns of} $B$,
%$\linh(B) = \{B t\colon t \in \RR^n \}$ is the \emph{linear hull spanned by columns of} $B$,
$\inth(B) = \{x = B t \colon t \in \ZZ^n\}$ is the \emph{lattice, spanned by columns of} $B$.
%If $D \subseteq \RR^n$, then the symbol $\linh(D)$ designates the linear hull, based on the points of $D$. The same is true for other types of the hulls.

\end{definition}

\subsection{The Smith and Hermite Normal Forms}\label{SNF_section}

Let $A \in \ZZ^{m \times n}$ be an integer matrix of rank $n$. It is a known fact (see, for example, \cite{SCHR98,STORH96}) that there exists a unimodular matrix $Q \in \ZZ^{n \times n}$, such that $A = \binom{H}{B} Q$, where $B \in \ZZ^{(m-n) \times n}$ and $H \in \ZZ_+^{n \times n}$ is a lower-triangular matrix, such that $0 \leq H_{i j} < H_{i i}$, for any $i \in \intint n$ and $j \in \intint (i-1)$. The matrix $\binom{H}{B}$ is called the \emph{Hermite Normal Form} (or, shortly, HNF) of the matrix $A$. Additionally, it was shown in \cite{FPT18} that $\|B\|_{\max} \leq \Delta(A)$ and, consequently, $\|\binom{H}{B}\|_{\max} \leq \Delta(A)$.

Let $A \in \ZZ^{m \times n}$ be again an integer matrix of rank $n$. It is a known fact (see, for example, \cite{SCHR98,STORS96,ZHEN05}) that there exist unimodular matrices $P \in \ZZ^{m \times m}$ and $Q \in \ZZ^{n \times n}$, such that $A = P \dbinom{S}{\BZero_{d \times n}} Q$, where $d = m - n$ and $S \in \ZZ_+^{n \times n}$ is a diagonal non-degenerate matrix. Moreover, $\prod_{i = 1}^{k} S_{ii} = \Delta_{\gcd}(k,A)$, and, consequently, $S_{ii} \mid S_{(i+1) (i+1)}$, for $i \in \intint (n-1)$. The matrix $\dbinom{S}{\BZero_{d \times n}}$ is called the \emph {Smith Normal Form} (or, shortly, SNF) of the matrix $A$.

Near-optimal polynomial-time algorithms for constructing the HNF and SNF of $A$ are given in \cite{STORS96,STORH96}. We denote by $T_{\SNF}(n)$ the arithmetical complexity of computing SNF for matrices $A\in\ZZ^{n\times n}$. 

\subsection{Valuations and polyhedra} \label{valuations_polyhedra_subs}

In this subsection, we follow to the monograph \cite{BARVBOOK} in the most of definitions and notations.

\begin{definition}
For a matrix $A \in \ZZ^{m \times n}$ and a vector $b \in \ZZ^{m}$, by $P_{\leq}(A,b)$ we denote the polyhedron $\{ x \in \RR^{n} : A x \leq b\}$ and by $P_{=}(A,b)$ we denote the polyhedron $\{ x \in \RR^{n}_+ : A x = b\}$.

The set of all vertices of a polyhedron $P$ is denoted by $\vertex(P)$.
\end{definition}

Let $V$ be a $d$-dimensional real vector space and $\Lambda \subset V$ be a lattice.

\begin{definition}
Let $A \subseteq V$ be a set. The \emph{indicator} $[A]$ of $A$ is the function $[A]\colon V \to \RR$ defined by
$$
[A](x) = \begin{cases}
1\text{, if }x \in A\\
0\text{, if }x \notin A.
\end{cases}
$$ 
The \emph{algebra of polyhedra} $\PC(V)$ is the vector space defined as the span of the indicator functions of all polyhedra $P \subset V$.
\end{definition}

\begin{definition}
A linear transformation $\TC \colon \PC(V) \to W$, where $W$ is a vector space, is called a \emph{valuation}. We consider only \emph{$\Lambda$-valuations} or \emph{lattice valuations} that satisfy
$$
\TC([P + u]) = \TC([P]) \quad \text{, for all rational polytopes }P \text{ and }u \in \Lambda,
$$ see \cite[pp. 933--988]{Mc93}, \cite{McS83}.
\end{definition}

\begin{remark}\label{quasipoly_interpolation} \cite{BARVSIMPLEX}. A general result of P.~McMullen \cite{Mc78} states that if $\nu(P) = \TC([P])$, for a lattice valuation $\TC$, $P \subset V$ is a rational polytope, $d = \dim P$, and $t \in \NN$ is a number, such that $t P$ is a lattice polytope, then there exist functions $\nu_i(P,\cdot) \colon \NN \to \CC$, such that 
\begin{gather*}
\nu(m P) = \sum\limits_{i=0}^d \nu_i(P, m)\, m^i, \quad\text{ for all } m \in \NN\text{, and}\\
\nu_i(P, m + t) = \nu_i(P, m), \quad\text{ for all } m \in \NN.
\end{gather*}

If we compute $\nu(q P)$, for $q = m,\, m + t,\, m+ 2 t,\, \dots,\, m + d t$, then we can obtain $\nu_i(P, m)$ by interpolation.
\end{remark}

We are mainly interested in two valuations, the first is counting the valuation $\EC([P]) = |P \cap \ZZ^d|$. Applying the result of P.~McMullen to $\EC([P])$, we conclude that
\begin{gather*}
|m P \cap \ZZ^d| = \sum\limits_{i=0}^d e_i(P, m) m^i, \quad\text{ for all } m \in \NN\text{, and}\\
e_i(P, m + t) = e_i(P, m), \quad\text{ for all } m \in \NN.
\end{gather*} The function on the right hand side is called \emph{Ehrhart quasi-polynomial of $P$} after E.~Ehrhart, who discovered
the existence of such polynomials \cite{EHR67}, see also \cite[Section~4.6]{STAN}.

The second valuation $\FC([P])$ that we are interested for is defined in the following theorem, proved by J.~Lawrence \cite{L91}, and, independently, 
by A.~Khovanskii and A.~Pukhlikov \cite{KP92}. We borrowed the formulation from  \cite[Section~13]{BARVBOOK}.
\begin{theorem}[\cite{KP92,L91}]\label{rational_gen_th}
Let $\mathcal{R}(\CC^d)$ be the space of rational functions on $\CC^d$ spanned by the functions of the type 
$$
\frac{\BX^v}{(1 - \BX^{u_1})\dots(1-\BX^{u_d})},
$$
where $v \in \ZZ^d$ and $u_i \in \ZZ^d \setminus \{0\}$, for $i \in \intint d$. Then there exists a linear transformation (a valuation)
$$
\FC \colon \PC(\QQ^d) \to \mathcal{R}(\CC^d),
$$
such that the following holds:
\begin{enumerate}
\item Let $P \subset \RR^d$ be a non-empty rational polyhedron without lines, and let $C$ be its recession cone. Let $C$ be generated by the rays $w_1, \dots, w_n$, for some $w_i \in \ZZ^d \setminus \{0\}$, and let us define 
$$
W_C = \bigl\{ \BX \in \CC^d \colon |\BX^{w_i}| < 1 \text{ for } i \in \intint n \bigr\}.
$$
Then, $W_C$ is a non-empty open set and, for all $\BX \in W_C$, the series
$$
\sum\limits_{m \in P \cap \ZZ^d} \BX^m
$$ converges absolutely and uniformly on compact subsets of $W_K$ to the function $f(P,\BX) = \FC([P]) \in \mathcal{R}(\CC^d)$.
\item If $P$ contains a line, then $f(P,\BX) = 0$.
\end{enumerate}
\end{theorem}

If $P$ is a rational polyhedron, we call $f(P,\BX)$ its \emph{short rational generating function}.

\section{Description of the results and related works}

The first polynomial-time in fixed dimension algorithm that constructs the short generating function of a polyhedron was proposed by A.~Barvinok in \cite{BARV93}. Further modifications and details were given in \cite{BARVBOOK,BARVPOM,BARVWOODS,DYERKAN}. An alternative approach was given in \cite{HIRAI,AltCounting}.

Our main result is Theorem \ref{main_th1} below. It states that the short rational generating function can be computed with a polynomial-time algorithm, even when the dimension is varying, but the other parameters are fixed. The proof can be found in Section \ref{proofs_sec}.
\begin{theorem}\label{main_th1}
Let $P$ be a rational polyhedron defined by the one of the following ways:
\begin{enumerate}
    \item $P = P_{\leq}(A,b)$, where $A \in \ZZ^{(n+k)\times n}$, $b \in \ZZ^{n+k}$, $\rank A = n$, and $\Delta = \Delta(A)$;
    \item $P = P_{=}(A,b)$, where $A \in \ZZ^{k \times n}$, $b \in \ZZ^{k}$, $\rank A = k$, $\Delta_{\gcd}(A) = 1$, and $\Delta = \Delta(A)$.
    % \item $P = \conv(P) + \cone(R)$, where $A = \begin{pmatrix} 
    % \BUnit & \BZero \\
    % P & R
    % \end{pmatrix} \in \ZZ^{(n+1) \times (n+k)}$, $\rank A = n+1$ and $\Delta = \Delta(A)$.
    % \item $P = \conv(P) + \cone(R)$, where $P \in \QQ^{n \times k_1}$, $R \in \QQ^{n \times k_2}$, $\dim P = n$. And for any vertex $v \in \vertex(P)$ the cone of feasible directions $C_v = \fcone(P,v)$ is generated by an integral matrix $U_v$, such that $\Delta(U_v)$ is bounded by a constant $\Delta$.
\end{enumerate}
Then, the short rational generating function $f(P,\BX)$, for $P \cap \ZZ^n$, can be computed with an algorithm having the arithmetical complexity 
\begin{equation}\label{main_complexity}
%O\left(d^{2 k + O(1)} + \; T_{\SNF}(d) \cdot d^{k + \log_2 \Delta} \right), 
O\left(T_{\SNF}(d) \cdot d^{k} \cdot d^{\log_2 \Delta} \right),
\end{equation} where $d = \dim P$ ($d = n$, for the case 1, and $d = n - k$, for the case 2). The short rational generating function has the form:
\begin{equation}\label{SGF_canonical}
f(P,\BX) = \sum\limits_{i \in I} \epsilon_i \frac{\BX^{v_i}}{(1 - \BX^{u_{i\,1}})\cdots(1- \BX^{u_{i\,d}})}.
\end{equation}

Here, $|I| \leq \binom{d+k}{k} \cdot d^{\log_2 \Delta}$; $\epsilon_i \in \{-1,1\}$, $v_i, u_{i\,j} \in \ZZ^n$, for $i \in I$ and $j \in \intint d$. 

%The algorithm is polynomial for fixed $k$ and $\Delta$.

 %The bit-encoding sizes of $v_i$ and $u_{i\,j}$ are bounded by a polynomial from an input size.
\end{theorem}

% \begin{theorem}\label{main_th1}
% Let $A \in \ZZ^{(n+k) \times n}$, $b \in \ZZ^{n+k}$, $\rank A = n$, $\Delta = \Delta(A)$ and $P = P_{\leq}(A,b)$. The short rational generating function $f(P,\BX)$ for $P \cap \ZZ^n$ can be computed by an algorithm with the arithmetic complexity 
% $$
% O\left(\binom{n+k}{k}^{2k+O(1)} + \; T_{\SNF}(n) \cdot \binom{n+k}{k} \cdot n^{\log_2 \Delta}\right).
% $$ The short rational generating function has the form:
% \begin{equation}\label{SGF_canonical}
% f(P,\BX) = \sum\limits_{i \in I} \epsilon_i \frac{\BX^{v_i}}{(1 - \BX^{u_{i\,1}})\cdots(1- \BX^{u_{i\,n}})}.
% \end{equation}

% Here, $|I| \leq \binom{n+k}{k} \cdot n^{\log_2 \Delta}$; $\epsilon_i \in \{-1,1\}$, $v_i, u_{i\,j} \in \ZZ^n$ for $i \in I$ and $j \in \intint n$. 

% %The algorithm is polynomial for fixed $k$ and $\Delta$.

%  %The bit-encoding sizes of $v_i$ and $u_{i\,j}$ are bounded by a polynomial from an input size.
% \end{theorem}

\begin{remark}
To make the text easier to read, we hide additive terms of the type $\poly(s)$ in $O$-notation, when we estimate the arithmetical complexity in our work. Here, $s$ denotes the input size.

For example, the formula $O(n^2)$ means $O\left(n^2 + \poly(s)\right)$. We note, that the computational complexity of this additional computation is not greater, than the computational complexity of computing the SNF or HNF of $A$.

The outputs and all intermediate variables occuring in proposed algorithms have polynomial-bounded bit-encoding size. Hence, these algorithms have polynomial bit-complexity, if the main parameters ($k$ and $\Delta$) are fixed.
\end{remark}

\begin{remark}\label{gcd_Smith_rm}
To simplify the formulae, in the formulation of Theorem \ref{main_th1}, for the case $P = P_{=}(A,b)$, we made the assumption that $\Delta_{\gcd}(A) = 1$. It can be done without loss of generality, because the original system $A x = b$, $x \geq \BZero$ can be polynomially transformed to the equivalent system $\widehat{A} x = \widehat{b}$,  $x \geq \BZero$ with $\Delta_{\gcd}(\widehat{A}) = 1$.

Indeed, let $A = P\, (S\,\BZero)\, Q$, where $(S\,\BZero) \in \ZZ^{k \times k}$ be the SNF of $A$, and $P \in \ZZ^{k \times k}$, $Q \in \ZZ^{n \times n}$ be unimodular matrices. We multiply rows of the original system $A x = b,\, x \geq \BZero$ by the matrix $(P S)^{-1}$. After this step, the original system is transformed to the equivalent system $(I_{k \times k}\,\BZero)\,Q\,x = b$, $x \geq \BZero$. Clearly, the matrix $(I_{k \times k}\,\BZero)$ is the SNF of $(I_{k \times k}\,\BZero)\,Q$, so its $\Delta_{\gcd}(\cdot)$ is equal $1$.
\end{remark}

\begin{remark}\label{main_diff_rem}
The main difficulty to prove Theorem \ref{main_th1} is impossibility to use the original \emph{sign decomposition procedure} introduced by A.~Barvinok in \cite{BARV93}, see also \cite[pp.~137--147]{BARVBOOK} and \cite[pp.~129--135]{ALG_IP_BOOK}. Let $A \in \ZZ^{n \times n}$, $\Delta = |\det(A)| > 0$ and $C = \cone(A)$. The sign decomposition procedure of A.~Barvinok decomposes the cone $C$ into at most $$
% n^{1 + \cfrac{\log_2 \log_2 \Delta}{\log_2 \frac{n}{n-1}}} = 
(\log \Delta)^{O(n \log n)}
$$ unimodular cones. The last formula is exponential on $n$. Hence, we need another variant of the sign decomposition that will give a polynomial on $n$ number of unimodular cones for fixed $\Delta$. Lemma \ref{sign_decomp_lm} will give a decomposition variant with at most $n^{\log_2 \Delta}$ unimodular cones. Another difficulty is that we need to solve the shortest vector sub-problems during execution of the Barvinok's decomposition algorithm. The last problem is NP-hard with respect to randomized reductions in general case.  

Additionally, there are some difficulties with transforming between ILPs in standard and canonical forms, while preserving small values of the parameter $k$. The key idea, helping to make the transform, is explained in Section \ref{minors_subs}.
\end{remark}

Using the Hadamard inequality, we can write a trivial complexity estimate in terms of $\Delta_1 = \Delta_1(A) = \|A\|_{\max}$, for the case $P = P_{=}(A,b)$. A better dependence from system's elements can be achieved for non-negative matrices using the aggregation technique, see Subsection \ref{knapsack_sec}.

\begin{corollary}\label{main_corr3}
Let $A \in \ZZ^{k \times n}$, $b \in \ZZ^{k}$, $\rank A = k$, $\Delta_1 = \Delta_1(A)$, $\Delta_{\gcd}(A) = 1$, and $P = P_{=}(A,b)$. The short rational generating function for $P \cap \ZZ^n$ can be computed with an algorithm having the arithmetical complexity 
$$
% O\left(\binom{d+k}{k}^{2 + O(1)} +\; T_{\SNF}(d) \cdot \binom{d+k}{k} \cdot d^{1/2\,k \log_2 k + k \log_2 \Delta_1} \right),
%O\left(n^{2 k + O(1)} +\; T_{\SNF}(d) \cdot d^{k ( 1 + 1/2\, \log_2 k) + k \log_2 \Delta_1} \right),
O\left( T_{\SNF}(d) \cdot d^{k(1 + 1/2 \log_2 k)} \cdot d^{k \log_2 \Delta_1} \right),
$$ where $d = \dim P = n-k$. 
\end{corollary}

\subsection{Application to counting lattice points in parametric polyhedra}

%\paragraph{Coefficients of Ehrhart quasipolynomial.}

If $P$ is the polyhedron defined in Theorem \ref{main_th1}, then there exists $t \in \ZZ_+$, such that $t P \subseteq \ZZ^n$ and $t \leq \Delta_{\lcm}(A) \leq \Delta!$. This gives us a straightforward way to calculate all coefficients of Ehrhart quasipolynomial for $P$ (up to their periodicity) with a polynomial-time algorithm, for any fixed $k$ and $\Delta$. Therefore, after preprocessing, that is polynomial, for any fixed $k$ and $\Delta$, a linear-time algorithm can be obtained to compute $|m P \cap \ZZ^n|$, for any given $m$. 

In other words, the following statement holds.
\begin{corollary}\label{Ehrhart_cor}
Let $P$ be a polyhedron defined in the formulation of Theorem \ref{main_th1}. Then, all coefficients of the Ehrhart quasipolynomial for $P$ can be computed with an algorithm having the arithmetical complexity $$O\left(T_{\SNF}(d) \cdot d^{k + 1} \cdot d^{\log_2 \Delta} \cdot \Delta_{\lcm}(A)\right).$$  \end{corollary}
\begin{proof}
Consider the first case: $P = P_{\leq}(A,b)$ and $d = n$. In Remark \ref{quasipoly_interpolation}, we discuss that the coefficients $\{e_j(P, m)\}$ of the Ehrhard quasipolynomial for $P$ can be computed using interpolation procedure with input values $|q P \cap \ZZ^d|$, for $q = m + i \cdot t$ and $i \in \intint[0]{d}$. Due to \cite{BARV93,BARVPOM}, see also \cite{BARVBOOK,BARVWOODS}, the values $|m P \cap \ZZ^d|$ can be computed in time that is proportional to length of the short rational generating function \eqref{SGF_canonical}. %As an interpolation algorithm we can use the Newton's method, whose arithmetic complexity can be bounded by $O(n^2)$.  
The period of the coefficients $\{e_j(P,q)\}$ is bounded by $\Delta_{\lcm}(A)$, all facts together give us the desired complexity bound.

Let us consider the second case: $P = P_{=}(A,b)$ and $d = n - k$. Lemma \ref{problem_transform_lm} and Corollary \ref{problem_transform_cor} give a polynomial-time algorithm to construct a $d$-dimensional polyhedron $\widehat{P} = P_{\leq}(\widehat{A}, \widehat{b}) \subseteq \RR^d$, such that there is a bijective map between $P \cap \ZZ^n$ and $\widehat{P} \cap \ZZ^d$. Moreover, it can be easily seen from the proof of Lemma \ref{problem_transform_lm} that the set $m P \cap \ZZ^n$ bijectivelly maps to the set $m \widehat{P} \cap \ZZ^d$, and, consequently, $|m P \cap \ZZ^n| = |m \widehat{P} \cap \ZZ^d|$. Hence, the Ehrhart quasipolynomials for $P$ and $\widehat{P}$ are coinciding, and we can apply the result of the first case to $\widehat{P}$.  
\end{proof}

\begin{remark}\label{Ehrhart_rm}
The previous Corollary gives a way to compute coefficients of the Ehrhart quasipolynomial of $\Delta$-modular simplices with a polynomial-time algorithm even in varying dimension. The problem with a close formulation was solved in \cite{BARVSIMPLEX}, where it was shown that last $k$ coefficients of the Ehrhart quasipolynomial of a rational simplex can be found with a polynomial-time algorithm, for any fixed $k$. Additionally, the paper \cite{BARVSIMPLEX} introduces an important concept of \emph{intermediate sums on polyhedra}. Given a polynomial $h(\cdot)$, a polyhedra $P \subseteq \RR^n$ and a lattice $L$, \emph{the intermediate sum $S^L(P,h)$} is defined by the formula
$$
S^L(P,h) = \sum_{x} \int_{P \cap (x + L)} h(y)\,dy,
$$ where the summation index $x$ runs over the projected lattice in $V/L$. The next significant step for the case of fixed dimension was done in \cite{REAL_EHRH}. It establishes existence of a polynomial time algorithm for the computation of the weighted
intermediate sum $S^L(P,h)$ of a simple polytope $P$ (given by its vertices), and the corresponding Ehrhart quasipolynomial $t \to S^L(t \cdot P,h)$, when the slicing space has fixed codimension and the weight depends only on a fixed number of variables, or has fixed degree.
\end{remark}

\begin{remark}
Although the resulting complexity bound is polynomial, for any fixed $\Delta$ and $k$, the dependence of $\Delta_{\lcm}(A)$ on $\Delta$ can be quite palpable. An alternative approach developed in \cite{CLAUSS,ParamVert,ParamCounting,GenCounting} gives a significantly better preprocessing complexity bound, and additionally gives an opportunity to work with more general parametric polytopes. The drawback is that the computation complexity of computing $\left|m P \cap \ZZ^n\right|$ given $m$ is not linear on $n$ and depends on $k$ and $\Delta$.
\end{remark}

Further in this subsection we follow to the definitions and notations from \cite{GenCounting}.

\begin{definition}
A \emph{step-polynomial} $g \colon \ZZ^n \to \QQ$ is a function of the form 
\begin{equation}\label{step_poly}
    g(\BX) = \sum\limits_{j = 1}^m \alpha_j \prod\limits_{k = 1}^{d_j} \lfloor a_{j\,k}^\top \BX + b_{j\,k} \rfloor,
\end{equation}
where $\alpha_j \in \QQ$, $a_{j k} \in \QQ^n$, $b_{j k} \in \QQ$. We say that the \emph{degree} of $g(\BX)$ is $\max_j\{d_j\}$ and \emph{length} is $m$. %, $\langle \cdot, \cdot \rangle$ is the standard inner product.

A \emph{piece-wise step-polynomial} $c \colon \ZZ^n \to \QQ$ is a collection of \emph{full-dimensional, rational, half-open polyhedra} $Q_i$ together with the corresponding functions $g_i \colon Q_i \cap \ZZ^n \to \QQ$, such that 
\begin{enumerate}
    \item[1)] $\QQ^n = \bigcup_i Q_i$, and $Q_i \cap Q_j = \emptyset$, for different $i,j$;
    
    \item[2)] $c(\BX) = g_i(\BX)$, for $\BX \in Q_i \cap \ZZ^n$;
    
    \item[3)] each $g_i$ is a step-polynomial.
\end{enumerate}
We say that \emph{degree} of $c(\BX)$ is $\max_i\{\deg(g_i)\}$ and  \emph{length} is $\max_i\{\length(g_i)\}$.
\end{definition}

\begin{remark}\label{half_open_step_poly_rm}
Our definition of a \emph{piece-wise step-polynomial} slightly differs of that given in \cite{GenCounting}, where the polyhedra $\{Q_i\}$ may be not full-dimensional and $\QQ^n = \bigcup_i \relint(Q_i)$, where $\relint(Q_i) \cap \relint(Q_j) = \emptyset$, for different $i,j$.
%The reason of this change is that the size of the resulting $\{Q_i\}$-decomposition has better dependence on the dimension, $n$.
The notion of a half-open decomposition was introduced in \cite{HALFOPEN}. By \emph{half-open polyhedra}, we mean a polyhedron, for which some of the facet-defining inequalities are strict.
\end{remark}

\begin{definition}
Let $P \subset \QQ^p \times \QQ^n$ be a rational polyhedron,  such that, for all $\BY \in \QQ^p$, the set $P_{\BY}$ is bounded, and we define the function $c_P \colon \ZZ^p \to \ZZ$ by
$$
c_P(\BY) = \left|P_{\BY} \cap \ZZ^n\right| = \left|\{x \in \ZZ^n \colon \tbinom{\BY}{x} \in P \}\right|.
$$
We call $P$ \emph{a parametric polytope}, because, if 
$$P = \{\tbinom{\BY}{x} \in \QQ^p \times \QQ^n \colon B \BY + A x \leq b\},$$ for some matrices $B \in \ZZ^{m \times p}$, $A \in \ZZ^{m \times n}$, and vector $b \in \ZZ^m$, then 
$$
P_{\BY} = \{x \in \QQ^n \colon A x \leq b - B \BY\}.
$$
\end{definition}

Due to results obtained in \cite{ParamCounting,GenCounting}, both piece-wise step-polynomial representation for $c_P(\BY)$ and its generating function $\sum_{\BY} c_P(\BY) \BX^{\BY}$ can be computed with polynomial-time algorithms, for any fixed $p$ and $n$.

Our second goal is to show that $c_P(\BY)$, represented by a piece-wise step-polynomial, can be computed with a polynomial-time algorithm even in varying dimension, for some other parameters to be fixed.
\begin{theorem}\label{main_th2}
Let $P \subset \QQ^p \times \QQ^n$ be a rational parametric polytope defined by one of the following ways:
\begin{enumerate}
    \item[1)] $P = \{\tbinom{\BY}{x} \in \QQ^p \times \QQ^n \colon B \BY + A x \leq b\}$, where $A \in \ZZ^{(n+k) \times n}$, $B \in \ZZ^{(n+k) \times p}$, $b \in \ZZ^{n+k}$, and $\rank A = n$;
    \item[2)] $P = \{\tbinom{\BY}{x} \in \QQ^p \times \QQ_+^n \colon B \BY + A x = b \}$, where $A \in \ZZ^{k \times n}$, $B \in \ZZ^{k \times p}$, $b \in \ZZ^{k}$, and $\rank A = k$.
\end{enumerate}

Then, a piece-wise step-polynomial that represents $c_P(\BY)$ can be computed with a polynomial-time algorithm, for $\Delta = \Delta(A)$, $k$, and $p$ being fixed. More precisely, the arithmetical complexity can be bounded by 
$$
d^{(k-1)  (p+1) + O(1)} \cdot d^{\log_2 \Delta},
$$ 
where $d = n$, for the first case, and $d = n - k$, for the second one. The degree of the resulting piece-wise step-polynomial is $d$. The length and number of pieces are bounded by $O(d^{k+1} \cdot d^{\log_2 \Delta})$ and $O(d^{(k-1) p})$, respectively.
\end{theorem}
The proof of Theorem \ref{main_th2} will be given in Section \ref{proofs_sec}.

\begin{remark}\label{step_poly_compl}
It can be seen from the proof of the previous Theorem, given in Section \ref{proofs_sec}, that the total number of faces of the pieces $\{Q_i\}$ of the step-polynomial $c_P(\BY)$ is bounded by $O(d^{k-1})$. Hence, we can construct an index data structure, such as a hash-table or a binary search tree, which will allow us to find the corresponding $Q_i$, for given $\BY \in \QQ^p$ in time $O(d^{k-1})$. After that the value $c_P(\BY)$ can be computed in time $O(d^{k+1} \cdot d^{\log_2 \Delta})$.
\end{remark}

\subsection{Some applications to the knapsack problem }\label{knapsack_sec}

The classical \emph{unbounded knapsack problem} can be formulated as:
\begin{gather}
c^\top x \to \max\notag\\
\begin{cases}
a^\top x = a_0\\
x \in \ZZ^n,
\end{cases}\label{knapsack_pr}
\end{gather} where $c,a \in \ZZ_{+}^n$, and $a_0 \in \ZZ_+$. The \emph{subset sum problem} is actually the problem to find any (not necessarily optimal) solution for the unbounded knapsack problem or to conclude that the set of feasible solutions is empty.

The \emph{multidimensional} variant of the unbounded knapsack problem can be formulated as: 
\begin{gather}
c^\top x \to \max\notag\\
\begin{cases}
A x = b\\
x \in \ZZ_+^n,
\end{cases}\label{multi_knapsack_pr}
\end{gather} where $A \in \ZZ_+^{k \times n}$, $b \in \ZZ_+^k$, and $c \in \ZZ_{+}^n$.

We note that the dynamic programming approach gives an algorithm to count integral points in multidimensional knapsack polytope with the arithmetical complexity $O(n \cdot \|b\|^k_{\infty})$, see, for example \cite{PFERC}. The memory requirements of this algorithm is $O(\|b\|^k_{\infty})$. An algorithm with a sufficiently better memory requirement was presented in \cite{HIRAI}.

Due to Corollary \ref{main_corr3}, the rational generating function for the knapsack polytope \eqref{knapsack_pr} can be found with an algorithm having the arithmetical complexity 
$O(T_{\SNF}(d) \cdot d^{1+\log_2 \|a\|_{\infty}})$, where $d = n-1$. For the multidimensional variant \eqref{multi_knapsack_pr}, the complexity becomes $O\left( T_{\SNF}(d) \cdot d^{k(1 + 1/2 \log_2 k)} \cdot d^{k \log_2 \Delta_1} \right)$, where $\Delta_1 = \Delta_1(A) = \|A\|_{\max}$ and $d = n-k$. Consequently, there exists a counting algorithm with the same arithmetical complexity bound.

\begin{remark}
In the case, when $\|c\|_{\infty}$ is also bounded, we count integer solutions of the problems \eqref{knapsack_pr} and \eqref{multi_knapsack_pr} that satisfy to $c^\top x \leq c_0$ or $c^\top x = c_0$ with a polynomial-time algorithm. 
% If $\|c\|_{\infty} \leq \Delta_1$, then the resulting complexity will be 
% $$
% O(T_{\SNF}(d) \cdot d^{2 + \log_2 k + k\log_2 (\beta+1)}).
% %O(d^{5 + \log_2 k} \cdot d^{k\log_2 (\beta+1)}).
% $$
\end{remark}

Additionally, let us consider the parametric version of the multidimensional knapsack polytope \eqref{multi_knapsack_pr} $P = \{\binom{\mathbf{b}}{x} \in \RR^k \times \RR_+^n \colon A x = \mathbf{b}\}$ parameterised by the right hand vector $\mathbf{b}$. Due to Theorem \ref{main_th2} and Remark \ref{step_poly_compl}, we can compute the corresponding piece-wise step polynomial $c_P(\mathbf{b})$ in time $d^{k^2 + O(1)} \cdot d^{\log_2 \Delta}$, where $\Delta = \Delta(A)$. Then, evaluation of $c_P(\mathbf{b})$, for given $\mathbf{b}$, can be done in time $O(d^{k+1} \cdot d^{\log_2 \Delta})$.

% Consider Problem \eqref{multi_knapsack_pr}. Let $\beta = \|b\|_{\infty}$. Without loss of generality we can assume that $\|A\|_{\max} \leq \beta$.
% According to \cite{AGRKNYAZ}, there exist aggregation coefficients $y \in \ZZ^k_+$, given by the formula
% $$
% y_i = \frac{(\beta+1)^k - (\beta+1)^{i-1}}{\beta}
% $$
% such that the set of feasible solutions to the problem \eqref{multi_knapsack_pr} coincides with the set of solutions to the system
% $$
% \begin{cases}
% y^\top A x = y^\top b\\
% x \in \ZZ^n.
% \end{cases}
% $$
% Since $\|y\|_\infty \leq (\beta + 1)^{k}/\beta$, we have $\|y^\top A\|_{\infty} \leq k (\beta + 1)^{k}$. Hence, by Theorem \ref{main_th1}, we obtain the following.
% \begin{corollary}\label{cor_multidim_knapsack}
% Let $P$ be the polyhedron of the problem \eqref{multi_knapsack_pr}. The short rational generating function for $P \cap \ZZ^n$ can be computed by an algorithm with the arithmetic complexity $$
% O(T_{\SNF}(d) \cdot d^{1 + \log_2 k + k\log_2 (\beta+1)}),
% %O(d^{4 + \log_2 k} \cdot d^{k \log_2 (\beta+1)}),
% $$ 
% where $d = \dim P = n-k$.
% \end{corollary}

% \begin{remark}
% Almost the same result can be achieved using the aggregation coefficients obtained in \cite{AGRBABA}. The work \cite{AGRVES} gives better aggregation coefficients than in \cite{AGRKNYAZ,AGRBABA}, but, to the best of our knowledge, it doesn't give any asymptotic improvements for Corollary~\ref{cor_multidim_knapsack}.  
% \end{remark}

\subsection{Counting lattice points in polytopes defined by convex hulls}

% \begin{theorem}\label{main_th3}
% Let $\PC$ be a rational polyhedron defined by the one of the following ways:
% \begin{enumerate}
%     \item $\PC = \conv(P) + \cone(R)$, where $A = \begin{pmatrix} 
%     \BUnit & \BZero \\
%     P & R
%     \end{pmatrix} \in \ZZ^{(n+1) \times (n+k)}$, $\rank A = n+1$ and $\Delta = \Delta(A)$.
%     \item $\PC = \conv(P) + \cone(R)$, where $P \in \QQ^{n \times k_1}$, $R \in \QQ^{n \times k_2}$, $\dim \PC = n$. And for any vertex $v \in \vertex(\PC)$ the cone of feasible directions $C_v = \fcone(\PC,v)$ is generated by an integral matrix $U_v$, such that $\Delta(U_v)$ is bounded by a constant $\Delta$. Let, additionally, $k = k_1 + k_2$.
% \end{enumerate}
% Then, the short rational generating function $f(\PC,\BX)$ for the set $\PC \cap \ZZ^n$ can be computed by an algorithm with the arithmetic complexity 
% \begin{equation}\label{main_complexity}
% %O\left(d^{2 k + O(1)} + \; T_{\SNF}(d) \cdot d^{k + \log_2 \Delta} \right), 
% O\left(\Delta \cdot n^{k+O(1)} \right).
% \end{equation} The short rational generating function has the form:
% \begin{equation}\label{SGF_canonical}
% f(P,\BX) = \sum\limits_{i \in I} \frac{ p_i(\BX) }{(1 - \BX^{u_{i\,1}})\cdots(1- \BX^{u_{i\,n}})}.
% \end{equation}

% Here, $|I| \leq \binom{d+k}{k-1}$, $u_{i\,j} \in \ZZ^n$ and $p_i(\BX)$ are polynomials with integer coefficients of the degree $n$ and length at most $\Delta$, for $i \in I$ and $j \in \intint n$.
% \end{theorem}

In this subsection, we prove a similar result to Theorem \ref{main_th1} that is stated for polyhedra defined by convex hulls of points. This Theorem is more like a note than an independent result, because the proof technique is straightforward. But, we include it there for the sake of completeness. 

For example, this Theorem can be applied to simplices with integral vertices, such that the matrix composed from vertex coordinates has bounded minors. The proof can be found in Section \ref{proofs_sec}. 

\begin{theorem}\label{main_th3}
Let $\PC$ be a $n$-dimensional rational polyhedron defined by the following way:
\begin{equation*}
    \PC = \conv(P) + \cone(R),
\end{equation*} where $P \in \ZZ^{n \times s_1}$, $R \in \ZZ^{n \times s_2}$. Let, additionally, $k = s_1 + s_2 - n$ and $\Delta = \Delta(P\,R)$.

Then, the short rational generating function $f(\PC,\BX)$ for the set $\PC \cap \ZZ^n$ can be computed with an algorithm having the arithmetical complexity 
\begin{equation}\label{main_complexity}
%O\left(d^{2 k + O(1)} + \; T_{\SNF}(d) \cdot d^{k + \log_2 \Delta} \right), 
O\left(n^{k+1} \cdot \Delta \right).
\end{equation} The short rational generating function has the form:
\begin{equation}\label{SGF_canonical}
f(P,\BX) = \sum\limits_{i \in I} \frac{ p_i(\BX) }{(1 - \BX^{u_{i\,1}})\cdots(1- \BX^{u_{i\,n}})}.
\end{equation}

Here, $|I| \leq \binom{d+k}{k-1}$, $u_{i\,j} \in \ZZ^n$ and $p_i(\BX)$ are polynomials with integer coefficients of degree $n$ and number of monomials at most $n \Delta$, for $i \in I$ and $j \in \intint n$.
\end{theorem}

\begin{remark} Let polyhedron $\PC$ be defined in the following way: $\PC = \conv(P) + \cone(R)$, where $P \in \QQ^{n \times s_1}$, $R \in \QQ^{n \times s_2}$, $\dim \PC = n$. For any vertex $v \in \vertex(\PC)$, the cone of feasible directions $C_v = \fcone(\PC,v)$ is generated by an integral matrix $U_v$, such that $\Delta(U_v)$ is bounded by a constant $\Delta$. Let, additionally, $k = s_1 + s_2 - n$. 

Using the same methods as in the proof of Theorems \ref{main_th1} and \ref{main_th3}, it can be easily shown that $f(\PC, \BX)$ can be computed with an algorithm having the arithmetical complexity $n^{k + O(1)} \cdot \Delta$.
\end{remark}

\subsection{Other related works}

Here we list some results that, in our opinion, are related to the topic under consideration. 

Let $A$ be an integer matrix and $b,c$ be integer vectors. By $\IP_{\leq}(A,b,c)$, we denote the problem $\max\{c^\top x \colon A x \leq b,\; x \in \ZZ^n\}$. By $\IP_{=}(A,b,c)$, we denote the problem $\max\{c^\top x \colon A x = b,\; x \in \ZZ_+^n\}$.

There are known some cases, when the $\IP_{\leq}(A,b,c)$ problem can be solved with a polynomial-time algorithm. It is well-known that if $\Delta(A) = 1$, then any optimal solution of the corresponding LP problem is integer. Hence, the $\IP_{\leq}(A,b,c)$ problem can be solved with any polynomial-time LP algorithm  (like in \cite{HGLOB95,KAR84,KHA80,NN94}).

The next natural step is to consider the \emph{bimodular} case, i.e. $\Delta(A) \leq 2$. The first paper that discovers fundamental properties of the bimodular ILP problem is \cite{VESCH09}. Recently, using results of \cite{VESCH09}, a strong polynomial-time solvability of the bimodular ILP problem was proved in \cite{AW17}.

Unfortunately, not much is known about the computational complexity of $\IP_{\leq}(A,b,c)$, for $\Delta(A) \geq 3$. V.N. Shevchenko \cite{SHEV96} conjectured that, for each fixed $\Delta = \Delta(A)$, the $\IP_{\leq}(A,b,c)$ problem can be solved with a polynomial-time algorithm. There are variants of this conjecture, where the augmented matrices $\dbinom{c^\top}{A}$ and $(A \, b)$ are considered \cite{AZ11,SHEV96}. A step towards deriving its complexity was done by Artmann et al. in \cite{AE16}. Namely, it has been shown that if the constraint matrix has additionally no singular rank sub-matrices, then the ILP problem with bounded $\Delta$ can be solved in polynomial time. The last fact was strengthened to a FPT-algorithm in \cite{FPT18}. Some interesting results about polynomial-time solvability of the boolean ILP problem were obtained in \cite{AZ11,BOCK14,GRIBM17,GRIBM18}.

F.~Eisenbrand and S.~Vempala \cite{EIS16} presented a randomized simplex-type linear programming algorithm, whose expected running time is strongly polynomial if all minors of the constraint matrix are bounded in the absolute value by a fixed constant. As it was mentioned in \cite{AW17}, due to E.~Tardos' results \cite{TAR86}, linear programs with the constraint matrices, whose all minors are bounded in the absolute value by a fixed constant, can be solved in strongly polynomial time. N.~Bonifas et al. \cite{BONY14} showed that any polyhedron, defined by a totally $\Delta$-modular matrix (i.e., a matrix, whose all rank order minors are $\pm\Delta$), has a diameter, bounded by a polynomial in $\Delta$ and the number of variables.

For the case, when $A$ is square, a FPT-algorithm can be obtained from the classical work of R.~Gomory \cite{GOM65}. Due to \cite{FPT18}, a FPT-algorithm exists for the case, when $A$ is almost square, e.g. $A$ has a small number of additional rows. It was shown in \cite{IntNumber} that, for fixed $A$, $c$, and varying $b$, the $\IP_{\leq}(A,b,c)$ problem can be solved by a FPT-algorithm with a high probability.

The existence of a FPT-algorithm with respect to $k$ and $\Delta$ for the $\IP_{=}(A,b,c)$ problem was shown in \cite{CONVILP}. A similar result for a more general problem with additional constraints in the form of upper bounds for variables was obtained in \cite{STEINITZILP}. 

Due to \cite{DistinctRowsNum}, the number of distinct rows in the system $A x \leq b$ can be estimated by $\Delta^{2 + \log_2 \log_2 \Delta} \cdot n + 1$, for $\Delta \geq 2$.

In \cite{GRIB13,GRIBV16}, it was shown that any lattice-free polyhedron $P_{\leq}(A,b)$ has a relatively small width, i.e. the width is bounded by a function that is linear in the dimension and exponential in $\Delta(A)$. Interestingly, due to \cite{GRIBV16}, the width of any empty lattice simplex, defined by a system $A x \leq b$, can be estimated by $\Delta(A)$.
%Class of polytopes considered in our paper (see Lemma \ref{diamond_lm}) also has this property.
In \cite{GRIBC16}, it has been shown that  the width of such simplices can be computed with a polynomial-time algorithm. In \cite{FPT18}, for this problem, a FPT-algorithm was proposed. In \cite{SVWidth19}, a similar FPT-algorithm was given for simplices, defined by the convex hull of columns of $\Delta$-modular matrices. We note that, due to \cite{SEB99}, this problem is NP-hard in the general case.

Important results about the proximity and sparsity of the LP, ILP, and mixed problems in the general case can be found in \cite{SupportIPSolutions,COGST86,ProximityUseSparsity,DistancesMixed}. Interestingly, due to \cite{ParametricFixedDim}, the maximum difference between the optimal values of the $\LP_{\leq}(A,b,c)$ and $\IP_{\leq}(A,b,c)$ problems over all right-hand sides $b \in \ZZ^m$, for which $\LP_{\leq}(A,b,c)$ is feasible, can be found with a polynomial-time algorithm, if the dimension is fixed.

In the case, when the parameter $\Delta$ is not fixed and the dimension parameter $n$ is fixed, the problems $\IP_{\leq}(A,b,c)$ and $\IP_{=}(A,b,c)$ can be solved with a polynomial time algorithm due to the famous work of Lenstra \cite{LEN83}. A similar result for general convex sets, defined by separation hyperplane oracle, is presented in \cite{DADDIS,DADFIXEDN}. Wider families of sets, induced by the classes of convic and discrete convic functions, defined by the comparison oracle, are considered in the papers \cite{CONVIC,CONVICD}. In \cite{CONVICM}, it was shown that the integer optimization problem in such classes of sets can be solved with an algorithm having the oracle complexity $O(n)^n$, which meets the same complexity bound as in \cite{DADDIS,DADFIXEDN}.

\subsubsection{Computing the simplex lattice width}

A.~Seb\"{o} \cite{SEB99} proved that the problem of computing the rational simplex width is NP-hard. A.~Y.~Chirkov and D.~V.~Gribanov \cite{GRIBC16} showed that the problem can be solved with a polynomial-time algorithm in the case, when the simplex is defined by a constraint matrix with bounded minors in the absolute value. Last result was improved to a FPT-algorithm in \cite{FPT18}. In \cite{SVWidth19}, a similar FPT-algorithm was given for simplices defined by convex hull of columns of $\Delta$-modular matrices. It was noted in \cite{GRIBV16} that the width of a integrally-empty $\Delta$-modular simplex is bounded by $\Delta$. In the current paper, we extend the class of polytopes with this property.

An interesting theory on demarcation of polynomial-time solvability and NP-completeness for graph problems is presented in \cite{MAL1,MAL2,MAL3,MAL4,MAL5}.

\section{Some auxiliary results}

\subsection{Minors of matrices with orthogonal columns}\label{minors_subs}

\noindent The following theorem was proved in \cite{PerpMatrix80}, see also \cite{SHEV96,PerpMatrix08}.
\begin{theorem}[\cite{PerpMatrix80}]\label{perp_matricies_th}
Let $A \in \ZZ^{n \times m}$, $B \in \ZZ^{n \times (n-m)}$, $\rank A = m$, $\rank B = n-m$, and $A^\top B = \BZero$. Then, for any $\BC \subseteq \intint n$, $|\BC| = m$, and $\NC = \intint{n} \setminus \BC$, the following equality holds:
$$
\Delta_{\gcd}(B)\,|\det A_{\BC *}| = \Delta_{\gcd}(A)\, |\det B_{\NC  *}|, \text{ where }\NC = \intint n \setminus \BC.
$$
\end{theorem}

\begin{remark}
Result of this Theorem was strengthened in \cite{PerpMatrix08}. Namely, it was shown that the matrices $A,B$ have the same diagonal of their Smith Normal Forms modulo of $\gcd$-like multipliers.
\end{remark}

The HNF can be used to solve systems of the type $A x = b$, see, for example, \cite{SCHR98}. In the following Lemma, we establish a connection between the minors of $A$ and the resulting solution matrix.

\begin{lemma}\label{problem_transform_lm}
Let $A \in \ZZ^{k \times n}$, $b \in \ZZ^k$, $\Delta_{\gcd}(A) = 1$, and $\rank A = k$. Let us consider the set $M = \{x \in \ZZ^n\colon A x = b\}$ of integer solutions of a linear equalities system. Then, there exist a matrix $B \in \ZZ^{n \times (n-k)}$ and a vector $r \in \ZZ^n$, such that $M = \inth(B) + r$ and $\Delta(B) = \Delta(A)$. The matrix $B$ and the vector $r$ can be computed with a polynomial-time algorithm.

\end{lemma}
\begin{proof}
The matrix $A$ can be reduced to the HNF. Let $A = (H\,\BZero)\, Q^{-1}$, where $H \in \ZZ^{k \times k}$, $(H\,\BZero)$ be the HNF of $A$, and $Q \in \ZZ^{n \times n}$ be a unimodular matrix. The original system is equivalent to the system $(H\,\BZero)\, y = b$, where $y = Q^{-1} x$. Hence, $y_{1:k} = H^{-1} b$ and components of $y_{(k+1):n}$ can take any integer values. Since $x = Q y$, we take $B = Q_{(k+1):n}$ and $r = Q_{1:k} H^{-1} b$.

We have $A B = \BZero$. The columns of $Q$ form a basis of the lattice $\ZZ^n$, so $\Delta_{\gcd}(B) = 1$. Hence, by Theorem \ref{perp_matricies_th}, we have $\Delta(B) = \Delta(A)$.
\end{proof}

We note that the map $f \colon \ZZ^{n-k} \to M$, defined by the formula $f(t) = B t + r$, is a bijection between the sets $\ZZ^{n-k}$ and $M$. Consequently, the polyhedron defined by a system in the canonical form can be transformed to an integrally-equivalent polyhedron in the standard form.
\begin{corollary}\label{problem_transform_cor}
Let $A \in \ZZ^{k \times n}$, $b \in \ZZ^k$, $\Delta_{\gcd}(A) = 1$, $\rank A = k$, $d = n - k$, and $P = P_{=}(A,b)$.
Consider the polyhedron $\widehat{P} = P_{\leq}(\widehat{A}, \widehat{b})$, where $\widehat{A} = -B \in \ZZ^{(d+k) \times d}$ and $\widehat{b} = r \in \ZZ^d$. The matrix $B$ and vector $r$ are taken from Lemma~\ref{problem_transform_lm}. Then,
\begin{enumerate}
    \item[1)] the map $x = \widehat{b} - \widehat{A} \widehat{x}$ is a bijection between the sets $P \cap \ZZ^n$ and $\widehat{P} \cap \ZZ^d$;
    \item[2)] moreover, $\Delta(A) = \Delta(\widehat{A})$;
    \item[3)] the matrix $\widehat{A}$ and the vector $\widehat{b}$ can be computed with a polynomial-time algorithm.
\end{enumerate}
\end{corollary}

\subsection{The algebra of polyhedra}

In this subsection, we mainly follow to \cite{BARVBOOK,BARVPOM}.
\begin{theorem}[Theorem~2.3 of \cite{BARVPOM}]\label{linear_map_eval_th}
Let $V$ and $W$ be finite-dimensional real vector spaces, and let $T \colon V \to W$ be an affine transformation. Then 
\begin{enumerate}
\item[1)] for every polyhedron $P \subset V$, the image $T(P) \subset W$ is polyhedron;
\item[2)] there is a unique linear transformation (valuation) $\mathcal{T} \colon \PC(V) \to \PC(W)$, such that 
$$
\mathcal{T}([P]) = [T(P)], \text{ for every polyhedron }P \subset V.
$$
\end{enumerate}
\end{theorem}

Let us fix a scalar product $(\cdot,\cdot)$ in $V$, just making $V$ Euclidean space.

\begin{definition}
Let $P \subset V$ be a non-empty set. The \emph{polar} $P^{\circ}$ of $P$ is defined by $$
P^\circ = \bigl\{ x \in V \colon (x,y) \leq 1 \; \forall y \in P\bigr\}.
$$
\end{definition}

\begin{definition}
Let $P \subset V$ be a non-empty polyhedron, and let $v \in P$ be a point. We define the \emph{tangent cone} of $P$ at $v$ by
$$
\tcone(P,v) = \bigl\{v+ y \colon v + \varepsilon y \in P, \; \text{ for some } \varepsilon > 0 \bigr\}.
$$
We define the \emph{cone of feasible directions} at $v$ by
$$
\fcone(P,v) = \bigl\{y \colon v + \varepsilon y \in P, \; \text{ for some } \varepsilon > 0 \bigr\}.
$$
Thus, $\tcone(P,v) = v + \fcone(P,v)$.
\end{definition}

\begin{remark}\label{dual_cone_generator_rm}
Let $A \in \RR^{m \times n}$, $b \in \RR^m$, and $P = P_{\leq}(A,b)$. Let, additionally, $v \in \vertex(P)$ and $J(v) = \{ j \colon A_{j *} v = b_{j}\}$. Then, from elementary theory of convex polyhedra it follows that
\begin{gather*}
\tcone(P,v) = \{ x \in V \colon A_{J(v) *} x  \leq b_{J(v)}\},\\
\fcone(P,v) = \{ x \in V \colon A_{J(v) *} x \leq \BZero \}, \\
\fcone(P,v)^\circ = \cone(A_{J(v) *}^\top).
\end{gather*}
\end{remark}

The following Theorem estimates the complexity of the construction of a triangulation.
\begin{theorem}\label{triang_complexity}
Let columns of a matrix $A \in \ZZ^{d \times (d+k)}$ generate a pointed full-dimensional cone $C$. Then, a triangulation of $C$, given by a collection of simple cones $C_i = \cone(B_i)$, where $B_i$ are $d \times d$ sub-matrices of $A$, can be computed with an algorithm having the arithmetical complexity
$
O(d^{k + 1}).
$
\end{theorem}
\begin{proof}
It is clear that constructing a triangulation for a pointed cone in $\RR^d$ is equivalent to constructing a triangulation for a point configuration in $\RR^{d-1}$.

The triangulation $T$ for the point configuration of $n$ points in $\RR^{d-1}$ of $\rank r < d$ can be computed in no more than $O(r(d-r)^2 |T| )$ operations  \cite[Lemma~8.2.2]{DELOERA} (see also \cite{SHEVGRUZD}, where another algorithm is proposed).

% The number $T$ of simplices in the triangulation can be bounded from above by famous McMullen's Upper Bound Theorem \cite{Mc70}:
% $$
% |T| \le \binom{n- \left \lfloor d/2 \right \rfloor }{n - d + 1} + \binom{n - \left \lfloor (d+1)/2 \right \rfloor }{n - d}.
% $$
In our case, $n=d+k$, $r=d$, and $|T| \leq \binom{d+k}{d} = \binom{d+k}{k} = O(d^k)$. Hence, the arithmetical complexity of the algorithm is $O(d k^2 d^k) = O(d^{k+1})$.
\end{proof}

\section{Sign decomposition of cones} \label{sign_decomposition_sec}

As it was explained in Remark \ref{main_diff_rem}, we can not use the original sign decomposition procedure of A.~Barvinok, proposed in \cite{BARV93}, to prove Theorems \ref{main_th1} and \ref{main_th2}. Here we introduce our variant of the sign decomposition.

\begin{lemma}\label{sign_decomp_lm}
Let $C = \cone(U)$, where $U \in \ZZ^{n \times n}$, $|\det U| = \Delta > 0$. Then, there exist unimodular cones $C_i = \cone(U_i)$, defined by unimodular matrices $U_i$, and values $\epsilon_i \in \{-1,1\}$, such that 
\begin{equation}\label{sign_decomp}
[C] = \sum\limits_{i \in I} \epsilon_i [C_i] \ldcmod,
\end{equation} where $|I| \leq n^{\log_2 \Delta}$.
The matrices $U_i$ and values $\epsilon_i$ can be computed with an algorithm having the arithmetical complexity $T_{\SNF}(n) \cdot n^{\log_2 \Delta}$.
\end{lemma}
\begin{proof}
Let  $U(j,b)$ be the matrix obtained from $U$ by replacing $j$-th column with the column $b$. Let $b = U t$, for some $t \in \RR^n$. We need the following two key formulae. The first one can be found in  \cite[Section~16]{BARVBOOK} or in \cite[p. 107]{BARVPOM}:
\begin{gather}
\bigl[\cone(U)\bigr] = \sum\limits_{i = 1}^n \epsilon_i \bigl[\cone\bigl(U(i,b)\bigr)\bigr] \ldcmod, \label{decomp_obs} \\
\text{where }  \epsilon_i = \begin{cases}
\phantom{+}1,\text{ if replacing $U_i$ with $b$ does not change the orientation;}\\
-1,\text{ if replacing $U_i$ with $b$ changes the orientation;}\\
\phantom{+}0,\text{ if $\det U(i,b) = 0$.}
\end{cases} \notag
\end{gather}

The second formula is trivial:
\begin{equation}\label{det_obs}
|\det U(i,b)| = |t_i| \cdot |\det U|. 
\end{equation} Let us assume that, for any matrix $A \in \ZZ^n$, we are able to find a vector $b = A t$, such that $b \in \ZZ^n$ and $0 < \|t\|_\infty \leq 1/2$. Then, we can apply decomposition \eqref{decomp_obs}, recursively. Due to \eqref{det_obs}, one has $|\det A(i,b)| \leq 1/2|\det A|$, and hence we need at most $\log_2 \Delta$ recursive steps to find the desired unimodular decomposition of the original cone $C$. Clearly, the number of cones in a such decomposition is bounded by $n^{\log_2 \Delta}$.

Let us show how to find a vector $b = A t$, such that $b \in \ZZ^n$ and $0 < \|t\|_\infty \leq 1/2$, for a given matrix $A$. The vector $t$ can be found as a solution of the following system
\begin{equation}\label{system1}
\begin{cases}
A t \equiv 0 \pmod 1\\
0 < \|t\|_\infty \leq 1/2.
\end{cases}
\end{equation}

Let $A = P^{-1} S Q^{-1}$, where $S$ is the SNF of $A$ and $P,Q$ are $n \times n$ unimodular matrices. All possible solutions of the system 
$$
A t \equiv 0 \pmod 1
$$
are given by the formula 
\begin{equation}\label{system2}
t = Q S^{-1}b,\text{ for }b \in \ZZ^n.
\end{equation}
Hence, as a result of the mapping $t \sigma \mapsto t$, where $\sigma = S_{nn}$, the system \eqref{system1} transforms to
\begin{equation}\label{system3}
\begin{cases}
A t \equiv 0 \pmod{ \sigma }\\
0 < \|t\|_\infty \leq \frac{\sigma}{2}\\
t \in \ZZ^n.
\end{cases}
\end{equation}

To find some solution of the system \eqref{system3}, we can do the following operations: 
\begin{enumerate}
    \item[1)] compute some nonzero solution of the first part of the system using the formula \eqref{system2}, for example, we can take $\widehat{t} =   Q S^{-1} (0, \dots, 0, \sigma)^\top = Q_{* n}$; 
    \item[2)] adjust components of $\widehat{t}$ by addition of $\pm \sigma$ until the inequality $\|\widehat{t}\|_\infty \leq \frac{\sigma}{2}$ holds.
\end{enumerate}
Then, the desired vector $b$ can be found by the formula $b = A \widehat{t} / \sigma$.

Clearly, the arithmetical complexity of the whole algorithm is $T_{\SNF}(n) \cdot n^{\log_2 \Delta}$.
\end{proof}

\begin{lemma}\label{sign_decomp_lm2}
In assumptions of Lemma \ref{sign_decomp_lm}, let $b = U t$ be some column of resulting matrices $U_i$ in the decomposition \eqref{sign_decomp} after $k$ recursive steps of the algorithm. Then $\|t\|_\infty \leq (3/2)^{k-1}$.
\end{lemma}
\begin{proof}
Induction on $k$. Base case: $k = 1$. Clearly, in this case $\|t\|_\infty = 1$, if $b$ is a column of the initial matrix $U$, and $\|t\|_\infty \leq 1/2$, if $b$ is a column that was added by the algorithm at the first step.

Induction step $k-1 \mapsto k$. Let 
$$
[C] = \sum\limits_{j \in J} \alpha_j \bigr[\cone\bigl(U(j,a^{(j)})\bigr)\bigr] \ldcmod
$$ 
be a decomposition after the first step of the algorithm, where $a^{(j)} \in \ZZ^n$ and $\alpha_j \in \{-1,1\}$. Without loss of generality we can assume that the vector $b$ is a result of applying $k-1$ recursive steps to the matrix $U(1,a^{(1)})$. Hence, by induction, $b = y_1 a^{(1)} + \sum_{i=2}^n y_i U_{* i}$, for some $y \in \QQ^n$, satisfying to $\|y\|_\infty \leq (3/2)^{k-2}$. We also know that $a^{(1)} = U z$, for some $z \in \QQ^n$, satisfying to $\|z\|_\infty \leq 1/2$. Hence, $b = U t = z_1 y_1 U_{* 1} + \sum_{i = 2}^n (z_i y_1 + y_i) U_{*i}$. 
$$
\text{Finally, }|t_j| = \begin{cases}
|z_1 y_1| < |y_1| \leq (3/2)^{k-2},\text{ for }j = 1;\\
|z_j y_1 + y_j| \leq 1/2 |y_1| + |y_j| \leq (3/2)^{k-1},\text{ for }j \in \intint[2]{n}.
\end{cases}
$$
\end{proof}

\begin{corollary}\label{sign_decomp_cor}
In assumptions of Lemma \ref{sign_decomp_lm}, let $b = U t$ be some column of resulting matrices $U_i$ in the complete decomposition \eqref{sign_decomp}. Then $\|t\|_\infty \leq 2/3 \cdot \Delta^{\log_2 3/2}$. Bit-encoding sizes of the matrices $U_i$ and all intermediate variables are bounded by a polynomial from bit-encoding size of the input matrix $U$.
\end{corollary}
\begin{proof}
The inequality $\|t\|_\infty \leq 2/3 \cdot \Delta^{\log_2 3/2}$ straightforwardly follows from the previous Lemmas \ref{sign_decomp_lm} and \ref{sign_decomp_lm2}. Consequently, $\|U_i\|_{\max} \leq 2/3 \cdot n \|U\|_{\max} \cdot \Delta^{\log_2 3/2}$ and, hence, the bit-encoding size of $U_i$ is bounded by a polynomial from the bit-encoding size of $U$. Clearly, the same is true for all intermediate variables.
\end{proof}

\section{Proofs of main results}\label{proofs_sec}

Now, we are ready to prove Theorems \ref{main_th1} and \ref{main_th2}. The proof of Theorem \ref{main_th1} consists of two parts corresponding to both cases of the polyhedron $P$ representation.

\subsection{Proof of Theorem \ref{main_th1}}
\subsubsection{Case 1: $P = P_{\leq}(A,b)$}
%\begin{proof}
In our proof, we follow to \cite{BARV93,BARVPOM}, and, especially, to \cite{BARVBOOK}. The following formula is actually the Brion's theorem \cite{BRION}. We borrow it from \cite[Section~6]{BARVBOOK}:
\begin{equation}\label{tcone_decomp}
[P] = \sum\limits_{v \in \vertex(P)} \bigl[\tcone(P,v)\bigr] \lmod.
\end{equation}

Let us fix a vertex $v \in \vertex(P)$ and consider a cone $C = \fcone(P,v)$. Remind that $\tcone(P,v) = v + C$. Due to Remark \ref{dual_cone_generator_rm}, we have $C^\circ = \cone(A_{J(v) *}^\top)$, where $J(v) = \{j \colon A_{j *} v = b_j\}$. We apply the decomposition (triangulation) of $C^\circ$ into simple cones $S_j$. Let $q_v$ be the total number of simple cones in this decomposition. Clearly, for $j \in \intint q_v$, we have $S_j = \cone(R_j)$, where $R_j$ are non-singular $n \times n$ integral sub-matrices of $A^\top$.  Since $|\det R_j| \leq \Delta$, we can apply the sign decomposition from Lemma \ref{sign_decomp}. After all these steps we will have the decomposition 
$$
[C^\circ] = \sum\limits_{i \in I} \epsilon_i [C_i] \ldcmod,
$$ where $C_i = \cone(U_i)$ for unimodular $U_i$ and $|I| \leq q_v \cdot n^{\log_2 \Delta} $.

Next, we use the \emph{duality trick}, see \cite[Remark~4.3]{BARVPOM}. Due to \cite[Theorem~5.3]{BARVBOOK} (see also \cite[Theorem~2.7 of]{BARVBOOK}), there is a unique linear transformation $\DC \colon \PC(V) \to \RR$, for which $\DC([P]) = [P^{\circ}]$. Consequently,
$$
[C] = \sum\limits_{i \in I} \epsilon_i [C_i^\circ] \lmod,
$$ where $C_i^\circ = \cone(B_i)$ and $B_i = -U_i^{-1}$. Since $U_i$ is unimodular, the matrix $B_i$ is also unimodular.

We decompose $\tcone(P,v)$ as follows:
\begin{equation}\label{tcone_decomp2}
[\tcone(P,v)] = v + C = \sum\limits_{i \in I} \epsilon_i \bigl[v + \cone(B_i)\bigr] \lmod.
\end{equation}

Now, we need to change the rational offset $v$ of the cone $v + \cone(B_i)$ to some integral offset $w_i$, such that the set of integral points of both cones stays the same. To this end, we follow to \cite[Section~14.3]{BARVBOOK}. Let a vector $t_i \in \QQ^n$ be defined by the equality $v = B_i t_i$. Define the vector $w_i = B_i \lceil t_i \rceil$. Then, $(v + \cone(B_i)) \cap \ZZ^n = (w_i + \cone(B_i))\cap \ZZ^n$, and, consequently,
\begin{equation}\label{vertex_rounding}
f\bigl(v + \cone(B_i), \BX\bigr) = f\bigl(w_i + \cone(B_i), \BX\bigr) = \BX^{w_i} f\bigl(\cone(B_i), \BX\bigr).
\end{equation}
Combining the previous formula and the formula \eqref{tcone_decomp2}, we have
\begin{equation}\label{tcone_decomp3}
f\bigl(\tcone(P,v), \BX\bigr) = \sum\limits_{i \in I} \epsilon_i  \BX^{w_i}  f\bigl(\cone(B_i), \BX\bigr).
\end{equation}

Let us fix $i$, and let $u_1,u_2, \dots, u_n$ be columns of the matrix $B_i$. We have (see \cite[Section~14.2]{BARVBOOK}) 
\begin{equation}\label{unimodular_generating_fun}
f(\cone(B_i), \BX) = \frac{1}{(1-\BX^{u_1}) (1-\BX^{u_2}) \cdots (1-\BX^{u_n})}.
\end{equation}

Combining the formulae \eqref{tcone_decomp}, \eqref{tcone_decomp3}, and \eqref{unimodular_generating_fun} together, we finally have
$$
f(P, \BX) = \sum\limits_{v \in \vertex(P)} \, \sum\limits_{i \in I_v} \epsilon^{(v)}_i \frac{\BX^{w_i^{(v)}}}{(1-\BX^{u^{(v)}_1}) (1-\BX^{u^{(v)}_2}) \cdots (1-\BX^{u^{(v)}_n})}.
$$
Since $\sum\limits_{v \in \vertex(P)} q_v \leq \binom{n+k}{n} = \binom{n+k}{k}$, the total number of terms in the resulting formula is bounded by $\binom{n+k}{k} \cdot n^{\log_2 \Delta}$.

Let us estimate the computational complexity of the resulting algorithm. Briefly, the algorithm consists of the following steps: 
\begin{enumerate}
    \item[1)] compute all vertices of $P$;
    \item[2)] for each $v \in \vertex(P)$, compute the decomposition (triangulation) of the cone $\cone(A_{J(v) *}^\top)$ into simple cones;
    \item[3)] for each simple cone, compute its sign decomposition using Lemma \ref{sign_decomp_lm}; 
    \item[4)] write down the resulting short generating function.
\end{enumerate}

Step 1) can be done just by enumerating all the bases of $A$, its arithmetical complexity can be estimated by $n^{O(1)} \cdot \binom{n+k}{k} = n^{k + O(1)}$, for fixed $k$. We note that this $n^{O(1)}$ term is not greater, than $T_{\SNF}(n)$. Since the number of simple cones in all triangulations is bounded by $\binom{n+k}{k} = O(n^k)$, then, using the same reasoning, as in Theorem \ref{triang_complexity}, we conclude that the arithmetical complexity of the step 2) is $O(n^{k+1})$.
% \begin{multline*}
% \sum\limits_{v \in \vertex(P)} n^{O(1)} \cdot q_v^2 \leq n^{O(1)} \left( \sum\limits_{v \in \vertex(P)} q_v \right)^2 = \\
% = n^{O(1)} \cdot \binom{n+k}{k}^2 = n^{2 k + O(1)}.
% \end{multline*}
Again, since the total number of simple cones is bounded by $O(n^k)$, then, by Lemma \ref{sign_decomp_lm}, the total arithmetical complexity of the step 3) is $O(T_{\SNF}(n) \cdot n^{k + \log_2 \Delta})$. The complexity of the resulting algorithm is the same as in the Theorem's formulation, so we are done.
%\end{proof}

\subsubsection{Case 2: $P = P_{=}(A,b)$}

%\begin{proof}
First of all, we apply Corollary \ref{problem_transform_cor} to get the $d$-dimensional polyhedron $\widehat{P} = P_{\leq}(\widehat{A}, \widehat{b}) \subseteq \RR^d$. We apply the result of {\it Case 1} for $\widehat{P}$ to get the decomposition
$$
[\widehat{P}] = \sum\limits_{v \in \vertex(\widehat{P})} \sum\limits_{i \in I_v} \epsilon_i^{(v)} \bigl[v + \cone(B_i)\bigr],
$$ 
where $B_i \in \ZZ^{d \times d}$ are unimodular matrices, and the number of terms in the formula is bounded by $\binom{d+k}{k} \cdot d^{\log_2 \Delta}$. 

% Due to Theorem 1.1 of \cite{Barv}, there is the unique linear transform $\TC \colon \PC(V) \to \RR$, for which $\TC([P]) = [T(P)]$, where $P$ is a polyhedron and $T$ is linear map

Next, we apply the map $x = \widehat{b} - \widehat{A} \widehat{x}$ to $\widehat{P}$. By Theorem \ref{linear_map_eval_th}, we have
\begin{equation}\label{part_decomp}
[P] = \sum\limits_{v \in \vertex(P)} \sum\limits_{i \in I_v} \epsilon_i^{(v)} [v + \cone(M_i)],
\end{equation}
where $M_i = - \widehat{A} B_i \in \ZZ^{n \times d}$.

Since the map $x = \widehat{b} - \widehat{A} \widehat{x}$ is a bijection between the sets $\{x \in \ZZ^n \colon A x = b\}$, $\ZZ^d$ and the unimodular matrices $B_i$ generate all points of the set $\cone(B_i) \cap \ZZ^d$, then the matrices $M_i$ have the same properties with respect to $\cone(M_i) \cap \ZZ^n$. Hence,
\begin{equation}\label{unimodular_generating_fun2}
f(\cone(M_i),\BX) = \frac{1}{(1-\BX^{u_1}) \cdots (1-\BX^{u_d})},
\end{equation} where $u_1, \dots, u_d$ are columns of $M_i$, for fixed $i \in I_v$.

Now, using the same trick, as in \eqref{vertex_rounding}, we find the vectors $w_i \in \ZZ^n$, such that 
\begin{equation}\label{vertex_rounding2}
f(v + \cone(M_i),\BX) = f(w_i + \cone(M_i),\BX) = \BX^{w_i} f(\cone(M_i),\BX).
\end{equation}

Combining the formulae \eqref{part_decomp}, \eqref{vertex_rounding2}, and \eqref{unimodular_generating_fun2} together, we finally have
$$
f(P, \BX) = \sum\limits_{v \in \vertex(P)} \, \sum\limits_{i \in I_v} \epsilon^{(v)}_i \frac{\BX^{w_i^{(v)}}}{(1-\BX^{u^{(v)}_1}) (1-\BX^{u^{(v)}_2}) \cdots (1-\BX^{u^{(v)}_d})}.
$$
Clearly, the total number of terms in the resulting formula is bounded by $\binom{d+k}{k} \cdot d^{\log_2 \Delta}$, and the total arithmetical complexity stays the same as in {\it Case 1} applied to the polyhedron $\widehat{P}$. So, the proof is complete.
%\end{proof}

\subsection{Proof of Theorem \ref{main_th2}}

\subsubsection{\bf Case 1: $P = \{\tbinom{\BY}{x} \in \QQ^p \times \QQ^n \colon B \BY + A x \leq b\}$}
In our proof, we partially follow to \cite{GenCounting}.

Due to \cite[Section~3]{CLAUSS}, we can find $p$-dimensional polyhedra $Q_i$, such that $\QQ^p = \bigcup_i Q_i$, $Q_i \cap Q_j$ are polyhedra of dimension lower than $d$, and, for any $\BY \in Q_i$, the polytopes $P_{\BY}$ will have a fixed set of vertices given by affine transformations of $\BY$. These $Q_i$, also called \emph{chambers}, will be the pieces of the resulting step-polynomial. Following to the proof of Lemma~3 in \cite{ParamCounting}, consider the hyperplanes in the parameter space $\QQ^p$ formed by the affine hulls of the $(p-1)$-dimensional intersections of pairs of the chambers. Let $s$ be the total number of such hyperplanes. Due to \cite{SPACESPLIT}, these hyperplanes divide the parameter space into at most $O(s^p)$ cells. The considered hyperplanes correspond to the projections of the generic $(p-1)$-dimensional faces of $\{\tbinom{\BY}{x} \in \QQ^p \times \QQ^n \colon B \BY + A x \leq b \}$ into the parameters space. Hence, $s \leq \binom{n+k}{n+1} = \binom{n+k}{k-1} = O(n^{k-1})$. Since a part of the cells form a subdivision of the chambers, the total number of the chambers can be bounded by $O(n^{(k-1) p})$.

The set of the chambers can be computed by Clauss and Loechner's algorithm \cite{CLAUSS}. The number of iterations of this algorithm is bounded by the number of parametric vertices that is bounded by $O(n^k)$. In each iteration, the number of  the performed operations is proportional to length of the list of pairs of regions and vertices. This length never decreases, and its final (maximal) value is  the number of the chambers. Hence, the total complexity of the chambers computation and parametric vertices computation can be estimated by $n^{(k-1)(1+ p) + O(1)}$. For each chamber $Q_i$, the Clauss and Loechner's algorithm returns the set of parametric vertices as a set of affine functions $T_{i\,1},\, T_{i\,2},\, \dots,\, T_{i\,m_i} \colon \QQ^p \to \QQ^n$. In other words, for $\BY \in Q_i$, all vertices of $P_{\BY}$ are $T_{i\,1}(\BY), T_{i\,2}(\BY), \dots, T_{i\,m_i}(\BY)$.

Now, let us fix a chamber $Q = Q_m$, for some $m$. Let $T_{1}(\BY), T_{2}(\BY), \dots, T_{t}(\BY)$ be the set of parametric vertices for this chamber. We note that for $\BY \in \inter(Q_m)$ all parametric vertices $T_{1}(\BY), T_{2}(\BY), \dots, T_{t}(\BY)$ are unique. But, for some $i \not= m$ and $\BY \in Q_i \cap Q_m$, this property disappears. To solve this problem, we take the set of all chambers $\{Q_i\}$ returned by the Clauss and Loechner's algorithm and transform them to half-open ones. Due to \cite{HALFOPEN}, this step can be done with an algorithm having the complexity, proportional to the number of the chambers and polynomial by dimension $n$.  

Again, let us fix $Q = Q_m$ and the corresponding set of parametric vertices $T_{1}(\BY), T_{2}(\BY), \dots, T_{t}(\BY)$. 
%For $i \in \intint t$ we set  $T_i(\BY) = M_i \BY + r_i$, where $M_i$ and $r_i$ are rational matrices and vectors returned by the Clauss and Loechner's algorithm. 
Using Theorem \ref{main_th1}, we achieve the decomposition
\begin{equation}\label{parametric_decomposition}
f(P_{\BY}, \BX) = \sum\limits_{i \in I} \epsilon_i \frac{\BX^{w_i(\BY)}}{(1 - \BX^{u_{i 1}}) (1 - \BX^{u_{i 2}}) \dots (1 - \BX^{u_{i n}}) },
\end{equation} where $w_i(\BY) = \sum_{j = 1}^n \lfloor L_{i j}(\BY) \rfloor u_{i j}$, for affine functions $L_{i j}$ (here, we again use the trick \eqref{vertex_rounding}), and $|I| \leq \binom{n+k}{k} \cdot n^{\log_2 \Delta}$. We note that $w_i(\BY)$ is a degree-$1$ step-polynomial, and the functions $L_{i j}(\BY)$ can be computed in polynomial time from the functions $T_{i}(\BY)$.

By Theorem \ref{rational_gen_th}, the finite sum $\sum\limits_{m \in P_{\BY} \cap \ZZ^n} \BX^m$ converges to $f(P_{\BY}, \BX)$, for $\BX \not= 1$. Hence, $c_P(\BY) = \lim\limits_{\BX \to \BUnit} f(P_{\BY}, \BX)$. To find this limit, we follow to Section~5 of the work \cite{BARVPOM}. First of all, we search for the vector $l \in \ZZ^n$, such that $(l, u_{i j}) \not= 0$, for all $i \in I$ and $j \in \intint n$. Due to \cite{BARVPOM}, the vector $l$ can be found with an algorithm having the arithmetical complexity $O(n^3 \, |I|)$ by taking points in the moment curve. Let $\xi_{i j} = (l, u_{i j})$ and $\eta_i(\BY) = (l,w_i(\BY))$, for $i \in I$ and $j \in \intint n$.
We note that $\eta_i(\BY)$ is a degree-$1$ step polynomial. Then, due to Formula (5.2.1) of \cite{BARVPOM}, we have
\begin{equation}\label{CP_formula}
c_{P}(\BY) = \sum\limits_{i \in I} \frac{1}{\xi_{i 1} \xi_{i 2} \dots \xi_{i n}} \sum\limits_{j = 0}^n \frac{\eta_i^j(\BY)}{j!} \toddp_{n-j}(\xi_{i 1}, \dots, \xi_{i n}),
\end{equation}
where $\toddp_j(\xi_1,\dots,\xi_n)$ are homogeneous polynomials of degree $j$, called \emph{$j$-th Todd polynomial} in $\xi_1,\dots,\xi_n$. These polynomials arise as the coefficients of the Taylor expansion 
$$
F(\tau; \xi_1,\dots,\xi_n) = \sum\limits_{i = 0}^{\infty} \tau^i \toddp_i(\xi_1,\dots,\xi_n)
$$
of the function 
$$
F(\tau; \xi_1,\dots,\xi_n) = \prod\limits_{i=1}^n \frac{\tau \xi_i}{ 1 - \exp(-\tau \xi_i)}
$$ that is analytic at $\tau = \xi_1 = \dots = \xi_n = 0$.

The formula \eqref{CP_formula} represents $c_{P}(\BY)$ as a degree-$n$ step-polynomial of the length $(n+1) |I| = O(n^{k+1} \cdot n^{\log_2 \Delta})$. To finish the proof, we need to show that the values $\toddp_{j}(\xi_1, \dots, \xi_n)$ can be computed with a polynomial-time algorithm, for $j \in \intint n$ and given $\xi_1,\dots,\xi_n$. Definitely, Theorem 7.2.8 from the book \cite[p. 137]{ALG_IP_BOOK} states that the value of $\toddp_{j}(\xi_1, \dots, \xi_n)$ can be evaluated by an algorithm that is polynomial by $j$, $n$ and the bit-encoding length of $\xi_1, \dots, \xi_n$.

% It is a known fact that the generating function of the Bernoulli numbers $\{B_j\}_{j=0}^{\infty}$ sequence is 
% \begin{equation*}
%     \frac{\tau}{e^\tau - 1} = \sum_{j = 0}^{\infty} \frac{B_j}{j !} \tau^j.
% \end{equation*}

% Hence, for $i \in \intint n$,
% \begin{equation*}
%     \frac{\tau \xi_i}{1 - \exp(-\tau \xi_i)} = \sum_{j = 0}^{\infty} \frac{(-\xi_i)^j\,B_j}{j !} \tau^j.
% \end{equation*}

% Consequently, for $m \in \intint n$, we have the recurrence
% \begin{equation}\label{recurrence_toddp}
% \toddp_m(\xi_1, \dots, \xi_n) = \sum_{j = 0}^m \toddp_{m-j}(\xi_1, \dots, \xi_{n-1}) \frac{(-\xi_n)^{j}\,B_{j}}{j !},
% \end{equation} starting from
% $$
% \toddp_m(\xi_1) = \frac{(-\xi_1)^m\,B_m}{m !}.
% $$
% %Due to \cite{BARVPOM}, the values $\toddp_j(\xi_1,\dots,\xi_n)$ can be computed in polynomial time for $j \in \intint n$.

% %Due to Theorem \ref{rational_gen_th}. The function $f(P_{\BY}, \BX)$ is analytic , to compute $c_P(\BY)$ we need to make the substitution $f(P_{\BY}, \BUnit)$. It can be done by substituting $\BUnit$ to each term of \eqref{parametric_decomposition}. Due to \cite{ParamCounting,GenCounting}, it can be done in polynomial time for each therm $$ \frac{\BX^{w_i(\BY)}}{(1 - \BX^{u_{i 1}}) (1 - \BX^{u_{i 2}}) \dots (1 - \BX^{u_{i n}}) }. $$ Hence, the substitution $f(P_{\BY}, \BUnit)$ can be done in time $n^{\log_2 \Delta + O(1)}$. 

% Clearly, the values $\toddp_{j}(\xi_1, \dots, \xi_n)$, for $j \in \intint n$, can be computed with a polynomial-time algorithm using the recurrence \eqref{recurrence_toddp}. 

The total complexity of the algorithm that consists of working with all chambers is $n^{(k-1)(1+p)+O(1)} \cdot n^{\log_2 \Delta}$, as it states in the theorem's definition.

\subsubsection{{Case 2:} $P = \{\tbinom{\BY}{x} \in \QQ^p \times \QQ_+^n \colon B \BY + A x = b\}$}

The reasoning here is the same as in the proof of Corollary \ref{Ehrhart_cor}. Using Lemma \ref{problem_transform_lm} and Corollary \ref{problem_transform_cor}, we just transform the polyhedron $P$ to the polyhedron $\widehat{P}$ that has the structure of {\it Case 1}.

\subsection{Proof of Theorem \ref{main_th3}}

\begin{proof}
Let $A = \begin{pmatrix} 
     \BUnit & \BZero \\
     P & R
     \end{pmatrix} \in \ZZ^{(n+1) \times (s_1 + s_2)}$ and $\hat \Delta = \Delta(A)$. Consider the homogenised cone $C = \cone(A)$. Clearly, 
$$
\hat \Delta \leq n \Delta \quad\text{and}\quad \PC = \{x \in \RR^n \colon \binom{1}{x} \in C\}.
$$

Hence, it can be seen that
\begin{equation}\label{homo_gen}
f(\PC,\BX) = \left. \frac{\partial}{\partial x_1} f(C, \BX) \right|_{x_1 = 0}.
\end{equation}

Let us triangulate $C$ to the set of simple cones $C_i = \cone(B_i)$, where $B_i$ are $(n+1) \times (n+1)$ sub-matrices of $A$. Let $m$ be total number of simple cones $\{C_i\}$. Clearly, $m \leq \binom{n+k}{n+1} = \binom{n+k}{k-1} = O(n^{k-1})$. Due to Theorem \ref{triang_complexity}, this step can be done in time $O(n^k)$. 

Now, we transform the cones $\{C_i\}$ to half-open ones $\{\tilde C_i\}$ using the method from \cite{HALFOPEN}. Due to \cite[Theorem~3]{HALFOPEN}, the complexity of this step is $O(m n^2)$. So, we have 
\begin{equation}\label{homo_cone_decomp}
    [C] = \sum_{i=1}^m [\tilde C_i].
\end{equation}

Let us fix some $i \in \intint m$ and consider $\tilde C = \tilde C_i$. Let $\tilde C$ have the following double description: 
\begin{gather*}
\tilde C = \left\{x \in \RR^{n+1} \colon (a^*_j, x) \leq 0\text{, for }j \in J_{\leq}\text{, and }(a^*_j, x) < 0\text{, for }j \in J_{<}\right\}\\
\tilde C = \left\{\sum\nolimits_{j = 1}^{n+1} a_j t_j \colon t_j \geq 0\text{, for }j \in J_{\leq}\text{, and }t_j > 0\text{, for }j \in J_{<}\right\},
\end{gather*}
where $J_{\leq} \cap J_{<} = \emptyset$, $J_{\leq} \cup J_{<} = \intint (n+1)$, and $a^*_i, a_j \in \ZZ^{n+1}$ are vectors with bi-orthogonality property
\begin{equation*}
     \begin{cases}
    (a^*_i, a_j) > 0,\text{ for } i = j\\
    (a^*_i, a_j) = 0, \text{ for } i \not= j
    \end{cases}.
\end{equation*}

Then, due to \cite[Section~3.1]{HALFOPEN}, 
\begin{equation}\label{half_open_cone_gen}
    f(\tilde C, \BX) = \frac{\sum_{m \in \paral \cap \ZZ^{n+1}} \BX^m}{(1 - \BX^{a_1})\dots (1- \BX^{a_{n+1}})},
\end{equation} where
\begin{equation*}
    \paral = \left\{\sum\nolimits_{j = 1}^{n+1} a_j t_j \colon 0 \leq t_j < 1\text{, for }j \in J_{\leq}\text{, and }0 < t_j < 1\text{, for }j \in J_{<}\right\}.
\end{equation*}

Clearly, $|\paral \cap \ZZ^{n+1}| \leq \hat \Delta$. Due to \cite[Lemma~9]{HALFOPEN}, the enumeration of points in $\paral \cap \ZZ^{n+1}$ can be done with $O(n \hat \Delta)$ operations. Hence, we need the same time to construct $f(\tilde C, \BX)$.

Combining the formulae \eqref{homo_cone_decomp} and \eqref{half_open_cone_gen}, we have 
$$
f(C, \BX) = \sum_{i \in I} \frac{p_i(\BX)}{(1 - \BX^{a_{i 1}})\cdots(1- \BX^{a_{i (n+1)}})},
$$ where $|I| \leq \binom{n+k}{k-1}$ and $p_i(\BX)$ are polynomials with integer coefficients of the degree $n+1$ and the number of terms is at most $\hat \Delta \leq n \Delta$. The arithmetical complexity of this computation is $O(m \cdot n^2 + m \cdot n \cdot \hat \Delta + n^k) = O(n^{k+1} \cdot \Delta)$.

To construct $f(\PC,\BX)$, we apply the formula \eqref{homo_gen}. Clearly, it can be done in linear time by length of the formula, and we finish the proof.

\end{proof}

\section*{Conclusion}

The paper considers a class of polyhedra $P$ defined by one of the following ways: 
\begin{enumerate}
\item[(i)] $P = \{x \in \RR^n \colon A x \leq b\}$, where $A \in \ZZ^{(n+k) \times n}$, $b \in \ZZ^{(n+k)}$ and $\rank A = n$,
\item[(ii)] $P = \{x \in \RR_+^n \colon A x = b\}$, where $A \in \ZZ^{k \times n}$, $b \in \ZZ^{k}$ and $\rank A = k$,
\item[(iii)] $P = \conv(A_1) + \cone(A_2)$, where $A = (A_1\,A_2)$, $A \in \ZZ^{n \times (n+k)}$ and $\dim P = n$.
\end{enumerate} 
We assume that all rank-order minors of $A$ are bounded by $\Delta$ in absolute values. It was shown that the short rational generating function for the set $P \cap \ZZ^n$ can be computed by a polynomial time algorithm for $k$ and $\Delta$ being fixed, even in varying dimension. Consequently, it gives polynomial time algorithm to compute $|P \cap \ZZ^n|$. The analogues results were proved for the parametric case.

There are some interesting questions for the future research:
\begin{enumerate}
    \item Is it possible to extend the presented results for the bounded knapsack polytope of the type $\{x \in \RR^n \colon A x = b,\, 0 \leq x \leq u\}$, or equivalently, for two-side bounded polytopes of the type $\{x \in \RR^n \colon b_1 \leq A x \leq b_2\}$? Here again $\Delta(A) = \Delta$ is fixed, and $A \in \ZZ^{k \times n}$ or $A \in \ZZ^{(n+k) \times n}$ respectively, for $k$ being fixed.
    
    \item With the same assumptions on input matrix $A$, is it possible to compute the intermediate sum or the real Ehrhart quasipolynomial for $P$ by a polynomial time algorithm? For details see Remark \ref{Ehrhart_rm}. 
\end{enumerate}

\section*{Acknowledgments}

The article was prepared within the framework of the Basic Research Program at the National Research University Higher School of Economics (HSE).

\medskip

\noindent
The authors thank the anonymous referees for their useful remarks that helped to make the text and proofs shorter and clearer.

\end{document}